\documentclass[11pt]{article}

\usepackage[T1]{fontenc}
\usepackage[utf8]{inputenc}
\usepackage{lmodern}
\usepackage{microtype}
\usepackage[margin=1in]{geometry}
\usepackage{amsmath,amssymb,amsthm,mathtools}
\usepackage{booktabs}
\usepackage{algorithm}
\usepackage[noend]{algpseudocode}
\usepackage{enumitem}
\usepackage{natbib}
\usepackage{xurl}
\usepackage[hidelinks]{hyperref}
\hypersetup{
  pdftitle={Exact Online Rank Recycling in Floyd's Uniform Subset Sampler},
  pdfauthor={Yingqi Zhang},
  pdfsubject={Exact uniform subset sampling and randomness recycling},
  pdfkeywords={uniform subset sampling, randomness recycling, Floyd sampling, random-bit complexity},
  pdfdisplaydoctitle=true,
  bookmarksopen=true
}

\allowdisplaybreaks
\newtheorem{theorem}{Theorem}

\newtheorem{corollary}[theorem]{Corollary}
\newtheorem{proposition}[theorem]{Proposition}
\theoremstyle{definition}
\newtheorem{definition}[theorem]{Definition}

\newcommand{\card}[1]{\left|#1\right|}
\newcommand{\Unif}{\operatorname{Unif}}
\newcommand{\rank}{\operatorname{rank}}
\newcommand{\Htwo}{H_2}
\newcommand{\bits}{\,\text{bits}}

\title{Exact Online Rank Recycling in Floyd's\\Uniform Subset Sampler}
\author{%
Yingqi Zhang\\
\small Department of Computer Science and Technology, Tsinghua University\\
\small Beijing, China\\
\small \texttt{zhangyq24@mails.tsinghua.edu.cn}}
\date{}

\begin{document}
\maketitle

\begin{abstract}
A uniformly random $m$-subset of $[n]=\{0,\ldots,n-1\}$ has entropy
$\log_2\binom{n}{m}$.  Standard without-replacement procedures often expose
an additional ordering coordinate that is absent from the returned set.  We
show that Floyd's subset sampler admits an exact round-local factorization of
this coordinate.  In round $r$, let $S$ be an $(r-1)$-subset of $[j]$, let
$T\sim\Unif([j+1])$, and let $S'$ be the result of Floyd's transition.  If
$D$ is the zero-based rank of the original draw $T$ in $S'$, then
$(S,T)\leftrightarrow(S',D)$ is a bijection between
$\binom{[j]}{r-1}\times[j+1]$ and $\binom{[j+1]}{r}\times[r]$.
Consequently, $S'$ and $D$ are independent and uniform on their respective
spaces.  The digit $D$ can therefore be merged immediately into a residual
uniform random state; an induction shows that the partial subset remains
independent of that state after every round.  For
$k=\min(m,n-m)$, the sampling phase uses $O(k\log k)$ time and $O(k)$
auxiliary space with an order-statistic tree; explicitly materializing a
complement incurs the unavoidable output cost.  The combinatorial layer
avoids binomial-coefficient arithmetic and recovers the complete $k!$
state-space factor exactly.  We also give a finite counterexample showing
that analogous immediate rank recycling in a partial Fisher--Yates array is
invalid because the unselected suffix retains a correlated ordering.  A
64-bit Rust implementation is checked by exhaustive state-space enumeration
for all $n\leq 8$ and by an entropy-accounting trace for choosing $20{,}000$
of $30{,}000$ items.  We make no claim of runtime superiority over existing
subset samplers.
\end{abstract}

\section{Introduction}
\label{sec:introduction}

Exact uniform sampling without replacement is a basic primitive in
randomized algorithms, simulation, experimental design, and lottery-style
selection.  Throughout, $[n]=\{0,\ldots,n-1\}$ and
$\binom{[n]}{k}$ denotes the family of all $k$-subsets of $[n]$.
The statistical requirement is that every $m$-subset be returned with
probability $1/\binom{n}{m}$.  In the random-bit model, however, exactness is
only one part of the specification: source randomness is itself a resource.
A uniform $m$-subset has Shannon entropy
\begin{equation}
    H_{n,m}=\log_2\binom{n}{m},
    \label{eq:subset-entropy}
\end{equation}
and any exact sampler driven by independent fair bits has expected source
cost at least $H_{n,m}$ \citep{KnuthYao1976}.  For $n=30{,}000$ and
$m=20{,}000$,
\begin{equation}
    H_{30{,}000,20{,}000}
    =H_{30{,}000,10{,}000}
    \approx 27{,}541.1978846043\bits.
\end{equation}
This corresponds to $3{,}442.6497$ bytes of entropy.  The integer $3{,}443$
is only the smallest whole number of bytes whose raw state space is large
enough to cover all possible subsets.

A partial Fisher--Yates shuffle generates an ordered $k$-sample through
\begin{equation}
    P(n,k)=\frac{n!}{(n-k)!}
          =\binom{n}{k}k!
    \label{eq:falling-factorization}
\end{equation}
equiprobable choice sequences.  Returning only the underlying set discards
the $k!$ ordering factor.  Once the procedure has terminated, the
factorization is elementary: an ordered sample is equivalent to an unordered
set together with a permutation of its elements.  Recycling this coordinate
\emph{during} the procedure is more delicate.  A digit is safe to reuse only
when it is independent of every retained state variable that can affect
future output.

This paper makes four scoped contributions.
\begin{enumerate}[leftmargin=*,itemsep=2pt,topsep=4pt]
    \item We express every round of Floyd's subset sampler
    \citep[Column~12]{Bentley2000} as an explicit product bijection
    $(S,T)\leftrightarrow(S',D)$, where $D$ is the rank of
    the original draw in the updated set.
    \item We use this bijection to prove exact subset uniformity and an
    invariant of independence between the partial subset and the residual
    uniform state after every immediate recycling step.
    \item We show why global counting alone is insufficient: immediate
    insertion-rank recycling in a partial Fisher--Yates array is biased, even
    though the completed ordered sample has a set--permutation
    factorization.
    \item We separate the lossless combinatorial transformation from the
    bounded-range state implementation, bound the latter's branch loss, and
    validate the claimed product structure by exact finite-state
    enumeration.
\end{enumerate}

The broader objective of entropy-efficient subset generation is not new.
Uniform-state recycling has been studied for changing ranges and general
online stochastic processes \citep{Mennucci2010,OmerPacher2014,DraperSaad2026},
and fixed-weight binary words---equivalently, uniform subsets---have been
addressed directly by \citet{BodiniDurand2026}.  Their method makes
Fisher--Yates insertions distinguishable and later deconstructs the induced
permutation, either after completion or at block boundaries.  The
contribution here is more local: Floyd's future-relevant combinatorial state
is only a set, and each round exposes a Cartesian rank factor that is already
independent of that complete retained state.  We do not claim novelty for
randomness recycling, Floyd sampling, or entropy-efficient subset generation
as such; the claim concerns this round-local factorization and its immediate
use.

Section~\ref{sec:uniform-states} states the uniform-state interface and its
range-sampling loss.  Section~\ref{sec:failed-fy} gives the Fisher--Yates
counterexample.  Section~\ref{sec:floyd} proves the Floyd rank bijection and
the sampler invariant.  Sections~\ref{sec:entropy}--\ref{sec:evaluation}
separate combinatorial accounting, finite-word implementation, and exact
validation, followed by related work and conclusions.

\section{Uniform random states and bounded draws}
\label{sec:uniform-states}

This section specifies a supporting interface rather than a contribution of
this paper.  The uniform-state model and the split, rejection, and recycling
operations are standard or adapted from prior randomness-recycling
constructions \citep{Mennucci2010,OmerPacher2014,DraperSaad2026}.  We include
Algorithm~\ref{alg:range} to state the independence contract used later and
to separate it from the paper's combinatorial contribution.

\begin{definition}[Uniform state]
For a positive integer $M$, a pair $(Z,M)$ is a uniform state if
$Z\sim\Unif([M])$.  When $M$ is random, this means
$Z\mid M\sim\Unif([M])$.  A uniform state is independent of an external
variable $Y$ when the pair $(Z,M)$ is independent of $Y$.
\end{definition}

The state retains $\log_2 M$ bits of randomness.  Two elementary bijections
are sufficient for the construction.

\paragraph{Appending a source digit.}
If $U\sim\Unif([A])$ is independent of $(Z,M)$, then
\begin{equation}
    (Z,U)\longmapsto AZ+U
    \label{eq:append}
\end{equation}
is a bijection from $[M]\times[A]$ to $[AM]$.  For byte input, $A=256$.

\paragraph{Merging a recycled digit.}
If $D\sim\Unif([r])$ is independent of $(Z,M)$, then
\begin{equation}
    (Z,D)\longmapsto rZ+D
    \label{eq:merge-digit}
\end{equation}
is a bijection from $[M]\times[r]$ to $[Mr]$.  We denote this operation by
$\operatorname{Recycle}(D,r)$.

\subsection{A recycled bounded-uniform sampler}
\label{subsec:bounded}

Algorithm~\ref{alg:range} returns an exact draw from $[N]$ and retains the
offset within a rejected region rather than discarding the entire state.  A
parameter $R\geq 2$ bounds the rejection probability at each attempted
split.  The idealized algorithm uses unbounded integers; finite-word behavior
is discussed in Section~\ref{sec:implementation}.

\begin{algorithm}[t]
\caption{Bounded uniform draw with a residual uniform state}
\label{alg:range}
\begin{algorithmic}[1]
\Require Uniform state $(Z,M)$; target size $N\geq 1$; source digits
$U\sim\Unif([A])$; rejection parameter $R\geq 2$
\Ensure $X\sim\Unif([N])$, independent of the updated residual state
\Function{Draw}{$N$}
    \Loop
        \While{$(M\bmod N)R\geq M$}
            \State draw $U\sim\Unif([A])$
            \State $(Z,M)\gets(AZ+U,AM)$
        \EndWhile
        \State $q\gets\lfloor M/N\rfloor$, $s\gets M-qN$
        \If{$Z<qN$}
            \State $X\gets Z\bmod N$
            \State $(Z,M)\gets(\lfloor Z/N\rfloor,q)$
            \State \Return $X$
        \Else
            \State $(Z,M)\gets(M-1-Z,s)$
        \EndIf
    \EndLoop
\EndFunction
\end{algorithmic}
\end{algorithm}

For a fixed pre-split modulus, write $M=qN+s$ with $0\leq s<N$.
On acceptance, the map
\begin{equation}
    Z\longmapsto (Z\bmod N,\lfloor Z/N\rfloor)
    \label{eq:accepted-split}
\end{equation}
is a bijection from $[qN]$ to $[N]\times[q]$.  On rejection,
$Z\mapsto M-1-Z$ is a bijection from $\{qN,\ldots,M-1\}$ to $[s]$.
Thus a rejected offset is again a valid uniform state.

\begin{proposition}[Bounded-draw contract]
\label{prop:draw-contract}
Let the input uniform state $(Z,M)$ be independent of an external variable
$Y$.  Assume that all source digits used by the call are mutually
independent, uniform on $[A]$, and jointly independent of $(Z,M,Y)$.
Algorithm~\ref{alg:range} terminates almost surely, returns
$X\sim\Unif([N])$ independently of $Y$, and leaves a residual uniform state
independent of $(X,Y)$.
\end{proposition}

\begin{proof}
At the beginning of each loop iteration, maintain the invariant that the
current state is uniform and independent of $Y$.  Appending a source digit
preserves this invariant by the bijection in \eqref{eq:append}.  At an
attempted split, conditional on the current modulus $M$, the accepting region
has size $qN$.  Conditional on acceptance, \eqref{eq:accepted-split} maps the
state bijectively to $[N]\times[q]$; hence the returned value is uniform and
independent of both the residual state and $Y$.  Conditional on rejection,
the reindexing $Z\mapsto M-1-Z$ produces a uniform state on $[s]$ and restores
the loop invariant.  Because the conditional law of the returned value is
the same uniform law for every terminal modulus and branch history, mixing
over those histories preserves independence from the final state and from
$Y$.

The expansion loop is finite: repeated appends multiply $M$ by $A$, whereas
$M\bmod N<N$.  Whenever a split is attempted, its rejection probability is
$s/M<1/R$.  Therefore the probability of at least $h$ consecutive rejections
is at most $R^{-h}$, which proves almost-sure termination.
\end{proof}

\subsection{Branch-state loss}
\label{subsec:branch-loss}

The within-branch transformations above are bijective, but the control-flow
branch itself is not encoded back into the pool.  For a split with
\begin{equation}
    \varepsilon=\frac{s}{M},
\end{equation}
the unaccounted state-space loss on the accepted and rejected branches is,
respectively,
\begin{align}
    \ell_{\mathrm{acc}}
      &=\log_2 M-\log_2(qN)
        =-\log_2(1-\varepsilon), \\
    \ell_{\mathrm{rej}}
      &=\log_2 M-\log_2 s
        =-\log_2\varepsilon.
\end{align}
The expected loss of one attempted split is exactly
\begin{equation}
    (1-\varepsilon)\ell_{\mathrm{acc}}
      +\varepsilon\ell_{\mathrm{rej}}
      =\Htwo(\varepsilon),
    \label{eq:binary-entropy-loss}
\end{equation}
where $\Htwo$ denotes binary entropy.  Algorithm~\ref{alg:range} attempts a
split only when $\varepsilon<1/R$.  Since $R\geq2$ and $\Htwo$ is increasing
on $[0,1/2]$, a geometric domination argument gives the per-call bound
\begin{equation}
    \mathbb{E}[\text{total branch loss}]
    \leq \frac{\Htwo(1/R)}{1-1/R}.
    \label{eq:range-loss-bound}
\end{equation}
The bound is intentionally simple rather than tight.  Its purpose is to
separate a tunable loss in bounded-range sampling from the lossless subset
transformation developed below.

\section{Why a global factorization does not justify online recycling}
\label{sec:failed-fy}

Equation~\eqref{eq:falling-factorization} shows that a completed ordered
sample decomposes into a subset and one of $k!$ internal orders.  It does not
follow that an order digit may be returned to the random-state pool before
the ordered sampler has finished.

Consider the following tempting modification of a partial Fisher--Yates
shuffle.  At round $i=0,\ldots,k-1$, draw $x\in[n-i]$, swap array positions
$i$ and $i+x$, let $d_i$ be the rank of the newly selected value among the
first $i+1$ selected values, and immediately recycle $(d_i,i+1)$.  The array
suffix is retained for subsequent swaps.  Although the completed sequence
$(d_0,\ldots,d_{k-1})$ encodes a permutation, an intermediate digit $d_i$ is
correlated with the current ordering of the suffix.  Feeding that digit back
into the pool changes the law of later draws.

A finite counterexample requires neither asymptotics nor a statistical test.
Take $n=6$, $k=3$, and an initial state $Z\sim\Unif([120])$.  The three range
splits are exact because the successive moduli are divisible by $6$, $5$,
and $4$.  Exhaustive enumeration gives the multiplicities in
Table~\ref{tab:bad-fy}; Appendix~\ref{app:counterexample} lists the
corresponding subsets.  Uniformity would require each of the
$\binom{6}{3}=20$ subsets to occur six times.

\begin{table}[t]
\centering
\caption{Output multiplicities under immediate insertion-rank recycling in a
partial Fisher--Yates array with $n=6$, $k=3$, and 120 equiprobable input
states.}
\label{tab:bad-fy}
\begin{tabular}{@{}ccc@{}}
\toprule
Multiplicity per subset & Number of subsets & Uniform multiplicity \\
\midrule
4 & 4 & 6 \\
6 & 12 & 6 \\
8 & 4 & 6 \\
\bottomrule
\end{tabular}
\end{table}

The counterexample isolates the relevant proof obligation.  A recycled digit
must be independent not only of the user-visible partial output but also of
every hidden state variable that can influence future output.  Floyd's set
transition eliminates the candidate-array ordering and admits exactly such a
local independence proof.

\section{Floyd sampling with online rank recycling}
\label{sec:floyd}

Let $k=\min(m,n-m)$.  We sample a uniform $k$-subset and, when
$k=n-m$, return its deterministic complement.  The working set is stored in
an order-statistic dictionary supporting membership, insertion, and rank.
For a finite ordered set $A$ and $x\in A$, write
\begin{equation}
    \rank_A(x)=\card{\{y\in A:y<x\}}
\end{equation}
for the zero-based rank of $x$.  Algorithm~\ref{alg:floyd-recycle} uses the
abstract draw contract of Proposition~\ref{prop:draw-contract}.

\begin{algorithm}[t]
\caption{Uniform subset sampling with online rank recycling}
\label{alg:floyd-recycle}
\begin{algorithmic}[1]
\Require Population size $n$; requested size $0\leq m\leq n$; uniform state
$(Z,M)$
\Ensure A uniform $m$-subset of $[n]$ and an updated residual state
independent of it
\State $k\gets\min(m,n-m)$
\State $S\gets\varnothing$
\For{$r=1,\ldots,k$}
    \State $j\gets n-k+r-1$
    \State $t\gets\Call{Draw}{j+1}$
    \If{$t\in S$}
        \State $S\gets S\cup\{j\}$
    \Else
        \State $S\gets S\cup\{t\}$
    \EndIf
    \State $d\gets\rank_S(t)$
    \Comment{rank of the original draw}
    \State $(Z,M)\gets(rZ+d,rM)$
    \Comment{$\operatorname{Recycle}(d,r)$}
\EndFor
\If{$k=m$}
    \State \Return $S$
\Else
    \State \Return $[n]\setminus S$
\EndIf
\end{algorithmic}
\end{algorithm}

The collision branch deliberately recycles the rank of the original draw
$t$, not the rank of the replacement value $j$.  If $t\in S$ before the
update, it remains in the set after $j$ is inserted.  Moreover, $j$ is larger
than every element of the previous set, so this insertion does not change the
rank of $t$.

\subsection{The round-local bijection}

Fix a round $r\in\{1,\ldots,k\}$ and set $j=n-k+r-1$.  At the beginning of
the round, $S$ is an $(r-1)$-subset of $[j]$.  For $t\in[j+1]$, define the
Floyd transition
\begin{equation}
    F(S,t)=
    \begin{cases}
        S\cup\{t\}, & t\notin S,\\
        S\cup\{j\}, & t\in S,
    \end{cases}
    \label{eq:floyd-transition}
\end{equation}
and let
\begin{equation}
    d=\rank_{F(S,t)}(t)
      =\card{\{x\in F(S,t):x<t\}}\in[r].
\end{equation}

\begin{theorem}[Floyd rank bijection]
\label{thm:local-bijection}
The map
\begin{equation}
    \Phi_{j,r}:
    \binom{[j]}{r-1}\times[j+1]
       \longrightarrow
    \binom{[j+1]}{r}\times[r],
    \qquad
    (S,t)\longmapsto\bigl(F(S,t),\rank_{F(S,t)}(t)\bigr)
    \label{eq:local-bijection}
\end{equation}
is a bijection.
\end{theorem}

\begin{proof}
We construct the inverse.  Let $(S',d)$ be in the codomain, and let $t$ be
the element of zero-based rank $d$ in the sorted set $S'$.  Define
\begin{equation}
    S=
    \begin{cases}
        S'\setminus\{j\}, & j\in S',\\
        S'\setminus\{t\}, & j\notin S'.
    \end{cases}
    \label{eq:local-inverse}
\end{equation}
In either case, $S$ is an $(r-1)$-subset of $[j]$.  If $j\notin S'$, then
$t\notin S$ and the forward transition inserts $t$.  If $j\in S'$ and
$t\neq j$, then $t\in S$, so the forward transition detects a collision and
inserts $j$.  Finally, if $t=j$, then $j\notin S$ and the forward transition
inserts $t=j$ directly.  Each case reconstructs $(S',d)$, and the
reconstruction is unique.
\end{proof}

The associated cardinality identity is
\begin{equation}
    \binom{j}{r-1}(j+1)=\binom{j+1}{r}r.
    \label{eq:counting-identity}
\end{equation}
The explicit inverse is stronger than this count: it identifies the precise
rank digit that can be recycled while preserving all combinatorial state
needed by later rounds.

\begin{corollary}[Exactness and output--pool independence]
\label{cor:uniformity}
Under the assumptions of Proposition~\ref{prop:draw-contract},
Algorithm~\ref{alg:floyd-recycle} returns each $m$-subset of $[n]$ with
probability $1/\binom{n}{m}$.  The returned subset is independent of the final
residual uniform state.
\end{corollary}

\begin{proof}
We prove by induction on $r$ that, before round $r$, the current set is
uniform over $\binom{[j]}{r-1}$ and independent of the pool.  The base case is
$S=\varnothing$.  For the inductive step, apply
Proposition~\ref{prop:draw-contract} with the current set as the external
variable.  The call returns $t\sim\Unif([j+1])$ and a residual pool
independent of $(S,t)$.  By Theorem~\ref{thm:local-bijection}, the pair
$(S',d)=\Phi_{j,r}(S,t)$ is uniform over
$\binom{[j+1]}{r}\times[r]$.  Thus $S'$ is uniform, $d$ is uniform on $[r]$,
and the two are independent.  The residual pool is independent of
$(S',d)$ because these are deterministic functions of $(S,t)$.  Merging $d$
by \eqref{eq:merge-digit} therefore yields a uniform state independent of
$S'$, completing the induction.

After $k$ rounds, $S$ is a uniform $k$-subset of $[n]$.  If $k=n-m$, the
complement map is a deterministic bijection from $k$-subsets to $m$-subsets,
so it preserves both uniformity and independence from the pool.
\end{proof}

\section{Entropy accounting}
\label{sec:entropy}

The $k$ Floyd rounds draw from ranges
\begin{equation}
    n-k+1,n-k+2,\ldots,n,
\end{equation}
whose product is $P(n,k)$.  Their recycled rank radices are
$1,2,\ldots,k$, whose product is $k!$.  Hence the net state-space factor
consumed by the combinatorial layer is
\begin{equation}
    \frac{P(n,k)}{k!}=\binom{n}{k}=\binom{n}{m}.
    \label{eq:net-combination-factor}
\end{equation}
Equivalently,
\begin{equation}
    \sum_{r=1}^{k}\log_2\frac{n-k+r}{r}
      =\log_2\binom{n}{k}.
\end{equation}
Thus the Floyd/rank transformation itself is lossless.  Any gap between
source entropy, output entropy, and retained-pool entropy arises in the
bounded-range implementation---or in an explicit finite-word reset---rather
than in the conversion from ordered choices to an unordered subset.

An exact divisible case makes the product structure concrete.  Let the
initial modulus be $M_0$, assume $P(n,k)$ divides $M_0$, and suppose no refill
or rejection occurs.  Repeated range splits and the local bijections then
give
\begin{equation}
    [M_0]
    \longleftrightarrow
    \binom{[n]}{k}\times
    \left[\frac{k!M_0}{P(n,k)}\right].
    \label{eq:exact-global-product}
\end{equation}
Equivalently, the second factor decomposes as
$[k!]\times[M_0/P(n,k)]$; the implementation folds these two coordinates
into one mixed-radix residual state.  Corollary~\ref{cor:uniformity} extends
exactness and output--pool independence to arbitrary refill and branch
histories.

\section{Complexity and implementation boundaries}
\label{sec:implementation}

Let $k=\min(m,n-m)$.  With a balanced order-statistic tree, each round
performs a constant number of $O(\log k)$ membership, insertion, and rank
operations.  The sampling phase therefore uses
\begin{equation}
    O(k\log k)\text{ time}
    \qquad\text{and}\qquad
    O(k)\text{ auxiliary space}.
\end{equation}
Traversing the sampled side in sorted order adds $O(k)$ time.  When the
sampled side is the complement of the requested output, a merge-style scan
of $[n]$ and the sorted excluded set materializes the requested subset in
$O(n)$ time; this includes the unavoidable cost of emitting a dense output.
If a library inserts every returned value into a separate balanced tree, the
output container adds $O(m\log m)$ time.  These representation costs are not
part of the Floyd transition itself.

The reference implementation uses a 64-bit pool state, 32-bit range sizes,
byte refills, $R=2^{32}$, and an order-statistic tree.  Pool arithmetic is
constant-time in the usual word-RAM model.  Three finite-width boundaries are
worth making explicit.

\begin{enumerate}[leftmargin=*,itemsep=3pt,topsep=4pt]
    \item If appending another byte would overflow the modulus, the
    implementation draws a fresh full-width state and discards the old pool.
    The replacement is still exactly uniform, but retained entropy is lost;
    this reset has no analogue in the unbounded-integer model.
    \item Rank recycling within Algorithm~\ref{alg:floyd-recycle} cannot
    overflow a $W$-bit modulus.  If a successful split uses working modulus
    $M_{\mathrm{work}}$ and range $j+1$, the residual modulus is
    $q=\lfloor M_{\mathrm{work}}/(j+1)\rfloor$.  Since $r\leq j+1$,
    \begin{equation}
        rq
        =r\left\lfloor\frac{M_{\mathrm{work}}}{j+1}\right\rfloor
        \leq M_{\mathrm{work}}<2^W.
    \end{equation}
    Moreover, for $Z\in[q]$ and $d\in[r]$,
    $0\leq rZ+d<rq$.  Thus neither the recycled modulus nor its state value
    requires an overflow fallback in the Floyd loop.
    \item The construction is not a randomness extractor.  It preserves and
    reallocates supplied entropy; it does not repair a biased, correlated,
    or predictable source.  Cryptographic deployment would additionally
    require an appropriate source, side-channel and constant-time analysis,
    and a separate security argument.
\end{enumerate}

\section{Exact validation and a large-instance accounting trace}
\label{sec:evaluation}

We validate the combinatorial product structure by finite state-space
enumeration rather than by frequency tests alone.  For every $n\leq 8$ and
$0\leq m\leq n$, the initial pool ranges over all $n!$ states.  Because
$P(n,k)$ divides $n!$, these tests require no refill, rejection, or reset and
therefore isolate the Floyd/rank transformation.  They verify that
\begin{enumerate}[leftmargin=*,itemsep=2pt,topsep=4pt]
    \item each of the $\binom{n}{m}$ subsets occurs exactly
    $m!(n-m)!$ times;
    \item every output--residual-value pair is unique; and
    \item the final residual modulus is $m!(n-m)!$.
\end{enumerate}
The second property checks the full Cartesian product, not merely the output
marginal.  For $n=6$ and $m=3$, for example, each of the 20 subsets is paired
once with every residual value in $[6]$.

For the motivating instance, let $n=30{,}000$ and $m=20{,}000$, so the
algorithm samples a $k=10{,}000$ element complement.  Table~\ref{tab:large}
reports deterministic modulus accounting along the no-rejection trace of the
64-bit implementation.  This trace has probability approximately
$0.9999996005$ under a uniform source.  It reads 3,449 bytes, retains a state
with $1{,}963{,}181{,}278{,}250{,}000$ possibilities, and leaves only
$5.76344\times10^{-7}$ bits unaccounted for by the output and the retained
pool.  ``Output-only efficiency'' is output entropy divided by source entropy;
``output-plus-pool efficiency'' also credits the final residual entropy.

\begin{table}[t]
\centering
\caption{Entropy accounting for choosing $20{,}000$ of $30{,}000$ items on
the no-rejection trace.  Integer state counts are exact; only displayed
logarithms are rounded.}
\label{tab:large}
\begin{tabular}{@{}lr@{}}
\toprule
Quantity & Value \\
\midrule
Subset entropy $\log_2\binom{30{,}000}{10{,}000}$
    & $27{,}541.1978846043$ bits \\
Source entropy read & $27{,}592$ bits (3,449 bytes) \\
Retained states & $1{,}963{,}181{,}278{,}250{,}000$ \\
Retained entropy & $50.8021148193$ bits \\
Unaccounted branch loss & $5.76344\times10^{-7}$ bits \\
Output-only efficiency & $99.815880997\%$ \\
Output-plus-pool efficiency & $99.9999999979\%$ \\
\bottomrule
\end{tabular}
\end{table}

This is an entropy trace, not a runtime benchmark.  In particular, it does
not establish that an order-statistic tree is faster than competing subset
samplers.  A comparative engineering study would need to control allocation,
cache behavior, output representation, and random-source cost separately.

\section{Related work}
\label{sec:related}

\paragraph{Exact discrete sampling and retained random state.}
\citet{KnuthYao1976} established the random-bit complexity framework for
exact discrete sampling and the classical single-sample upper bound of
$H+2$ bits.  \citet{Lumbroso2013} introduced the Fast Dice Roller for exact
uniform integer generation, and \citet{SaadEtAl2020} developed the Fast
Loaded Dice Roller for general finite distributions.  Reuse of a bounded
uniform register appears in scaling and fair-die algorithms
\citep{Mennucci2010,OmerPacher2014}.  The interval algorithm of
\citet{HanHoshi1997} gives a general arithmetic-coding-like construction for
exact simulation of stochastic processes.  Most directly,
\citet{DraperSaad2026} develop an online theory of randomness recycling based
on uniform-state merge and split operations, with applications that include
changing-range draws and Fisher--Yates permutations.  These works supply the
state-pool layer used here.  The additional question in this paper is which
coordinate of a without-replacement process is independent of its complete
retained combinatorial state and can therefore be recycled immediately.

\paragraph{Enumerative generation of subsets.}
A uniform $k$-subset can be generated globally by drawing an integer in
$[\binom{n}{k}]$ and unranking it.  This is the random-generation counterpart
of enumerative source coding and combinatorial ranking
\citep{Cover1973,NijenhuisWilf1978,Ryabko2006}.  Such methods expose the
global bijection between $\binom{[n]}{k}$ and $[\binom{n}{k}]$ directly and
can attain the information-theoretic state-space size.  They generally use
global rank/unrank logic and arithmetic on binomial coefficients or
comparably wide integers.  The present construction is not a new ranking
order: it factorizes a standard local transition, uses only the successive
ranges $n-k+1,\ldots,n$ and radices $1,\ldots,k$, and never materializes
$\binom{n}{k}$.

\paragraph{Sampling without replacement.}
Fisher--Yates permutation generation is classical
\citep{Durstenfeld1964,Knuth1997}.  The set transition commonly called
Floyd's algorithm is presented by \citet[Column~12]{Bentley2000}, who
attributes it to Robert Floyd; modern treatments discuss it alongside other
simple subset samplers and output-order variants \citep{Ting2021}.  We claim
neither transition as new.  The counterexample in Section~\ref{sec:failed-fy}
does not contradict the offline Fisher--Yates factorization: it shows only
that an insertion-rank digit is unsafe to reuse while a correlated candidate
suffix remains future-relevant.

\paragraph{Related fixed-weight generation.}
\citet{BodiniDurand2026} study the same output family as fixed-Hamming-weight
binary words concurrently.  Their sparse construction performs distinguishable
Fisher--Yates insertions, interprets the resulting labels as a uniform
permutation, and deconstructs that permutation into independent bits.
Recycling after completion gives expected net cost
$\log_2\binom{n}{k}+O(k)$ in their polynomially sparse regime; a multi-stage
variant reinjects permutation entropy at block boundaries and obtains a
$(1+\varepsilon)$-type guarantee while retaining linear time for an
explicitly materialized binary word.  The two approaches are complementary.
Floyd's transition retains no candidate permutation,
Theorem~\ref{thm:local-bijection} supplies an exact product factor in every
round, and the rank digit can be merged immediately while remaining
independent of the updated set.  At the combinatorial layer this recovers the
full $k!$ factor exactly; finite-word range splitting is accounted for
separately.

\section{Conclusion}
\label{sec:conclusion}

Uniform subset sampling illustrates a general constraint on online
randomness recycling: a coordinate that is redundant in the completed output
need not be safe to reuse during the computation.  The retained suffix order
in partial Fisher--Yates gives a finite counterexample.  Floyd's transition
avoids that hidden state and admits a round-local bijection between the
previous set and a bounded draw, on one side, and the updated set and a rank
digit, on the other.

The bijection yields an exact sampler whose partial subset remains
independent of the residual random state after every round.  It recovers the
complete $k!$ ordering factor without binomial-coefficient arithmetic and
separates this lossless combinatorial transformation from the independently
quantifiable loss of bounded-range state management.  More broadly, a sound
online-recycling argument must track all future-relevant internal state, not
only the final user-visible output.

\section*{Software artifact}
The implementation analyzed in this paper is \texttt{entropy-pool} version
1.0.0 \citep{ZhangEntropyPool2026}, identified by Git tag
\texttt{v1.0.0} and commit
\texttt{50b4b7ce508c28e06b70d42270201ce6cd186468}.  The tagged source and the
exact finite-state tests are available in the
\href{https://github.com/yingqi-z20/entropy-pool/tree/v1.0.0}{GitHub v1.0.0 tree};
the same release is distributed through
\href{https://crates.io/crates/entropy-pool/1.0.0}{crates.io}.

\bibliographystyle{plainnat}
\bibliography{references}

\appendix

\section{Exact three-of-six counterexample}
\label{app:counterexample}

For reproducibility, this appendix fixes the convention used in
Section~\ref{sec:failed-fy}.  Initialize $a=[0,1,2,3,4,5]$ and
$(Z,M)=(z,120)$.  For $i=0,1,2$:
\begin{enumerate}[leftmargin=*,itemsep=1pt,topsep=3pt]
    \item split $(Z,M)$ exactly into $x=Z\bmod(6-i)$ and
    $(Z,M)=(\lfloor Z/(6-i)\rfloor,M/(6-i))$;
    \item swap $a_i$ and $a_{i+x}$;
    \item let $d$ be the zero-based rank of $a_i$ in
    $\{a_0,\ldots,a_i\}$; and
    \item set $(Z,M)=((i+1)Z+d,(i+1)M)$.
\end{enumerate}
Enumerating $z=0,\ldots,119$ gives the exact nonuniform counts
\begin{center}
\begin{tabular}{@{}cl@{}}
\toprule
Count & Subsets \\
\midrule
8 & $\{0,3,4\}$, $\{1,2,3\}$, $\{1,2,5\}$, $\{1,4,5\}$ \\
4 & $\{0,2,3\}$, $\{0,2,5\}$, $\{0,4,5\}$, $\{1,3,4\}$ \\
6 & all remaining 12 subsets \\
\bottomrule
\end{tabular}
\end{center}
The discrepancy occurs without range rejection, so it is purely a failure of
the attempted combinatorial recycling.

\section{Numerical details for the large trace}
\label{app:numerics}

Starting from one source byte gives modulus $M_0=256$.  In round $i$, let
$N_i$ be the requested range.  Bytes are appended until
$(M\bmod N_i)/M<2^{-32}$; on the successful branch, the split replaces $M$ by
$\lfloor M/N_i\rfloor$, after which rank recycling multiplies it by the round
number.  Along the no-rejection trace, this recurrence reads a total of 3,449
bytes and ends at
\begin{equation}
    M_{\mathrm{final}}=1{,}963{,}181{,}278{,}250{,}000.
\end{equation}
The accumulated accepted-branch loss is
\begin{equation}
    -\sum_{i=1}^{10{,}000}\log_2(1-\varepsilon_i)
      =5.763441543123\times10^{-7}\bits,
\end{equation}
which agrees with
\begin{equation}
    8\cdot3449
    -\log_2\binom{30{,}000}{10{,}000}
    -\log_2 M_{\mathrm{final}}
\end{equation}
up to the displayed rounding.  The trace probability is
$\prod_i(1-\varepsilon_i)\approx0.9999996005087542$.

\end{document}